\documentclass[conference,letterpaper,onecolumn]{IEEEtran}
\usepackage[margin=0.7in,footskip=0.25in]{geometry}
\usepackage[utf8]{inputenc} 
\usepackage[T1]{fontenc}
\usepackage{url}

\usepackage{bm}
\usepackage{bbm}
\usepackage{ifthen}
\usepackage{cite}
\usepackage{enumerate,enumitem}
\usepackage[cmex10]{amsmath} 
\usepackage[normalem]{ulem}

\usepackage{hyperref}

\usepackage{changepage,mathrsfs}
\usepackage{pgfplots}
\pgfplotsset{width=10cm,compat=newest}
\usepgfplotslibrary{units}
\usepackage{pgfplotstable}
\usepackage[T1]{fontenc}
\usepackage{centernot}
\usepackage{amsmath,amssymb,amsthm,amsfonts}
\usepackage{tikz,amsfonts}
\usepackage{graphicx}
\usepackage{adjustbox}
\usepackage{verbatim}
\usepackage{color}
\usepackage{pgffor}
\usepackage{cleveref}
\usepackage{array}
\usepackage{proba}
\usepackage{bm}

\theoremstyle{plain}
\newtheorem{thm}{Theorem}
\newtheorem{lemma}[thm]{Lemma}
\newtheorem{definition}{Definition}
\newtheorem{rem}{Remark}

\DeclareMathAlphabet      {\mathbfit}{OML}{cmm}{b}{it}

\newcommand{\bY}{\mathbf{Y}}
\newcommand{\by}{\mathbfit{y}}
\newcommand{\bX}{\mathbf{X}}
\newcommand{\bx}{\mathbfit{x}}

\newcommand{\unif}{\textnormal{Unif}}
\allowdisplaybreaks
\newtheorem{conj}{Conjecture}

\title{A Differential Equation Approach to the Most-Informative Boolean Function Conjecture}
\author{%
  \IEEEauthorblockN{Zijie Chen, Amin Gohari, and Chandra Nair}
  \IEEEauthorblockA{Department of Information Engineering\\
The Chinese University of Hong Kong\\
Sha Tin, NT, Hong Kong\\
\{zijie,agohari,chandra\}@ie.cuhk.edu.hk}
}

\begin{document}
\maketitle
\begin{abstract}
    We study the most-informative Boolean function conjecture using a differential equation approach. This leads to a formulation of a functional inequality on finite-dimensional random variables. We also develop a similar inequality in the case of the Hellinger conjecture. Finally, we conjecture a specific finite dimensional inequality that, if proved, will lead to a proof of the Boolean function conjecture in the balanced case. We further show that the above inequality holds modulo four explicit inequalities (all of which seems to hold via numerical simulation) with the first three containing just two variables and a final one involving four variables.
\end{abstract}

\section{Introduction}
Let $\mathbb{H}^n = \{-1,1\}^n$  denote the $n$-dimensional Boolean Hypercube centered at $\bf{0}$. The most informative Boolean function conjecture states the following:
\begin{conj}[\hspace{1sp}\cite{kumar2013Boolean}]
\label{conj:mi}
    Let $\bX\sim\unif{(\mathbb{H}^n)}$ be a random variable distributed uniformly on the Boolean Hypercube and $\bY$ be received after passing each bit of $\bX$ through a BSC channel with cross-over probability $p = \frac{1 - \rho}{2} \in [0,1]$. Let $f:\mathbb{H}^n\mapsto \mathbb{H}$ be a Boolean function that maps a Boolean sequence to a binary value. Then the following inequality holds:
    \begin{align*}
        I(f(\bX);\bY) \leq 1 - H_2\left(p\right),
    \end{align*}
    where $H_2(p) = -p\log_2(p) - (1-p)\log_2(1-p)$ is the binary entropy function.
\end{conj}

A weaker form of the conjecture  was studied in  
\cite{anantharam2013hypercontractivity}.
Samorodnitsky \cite{7498615} demonstrated the existence of a positive constant $\rho_0$ such that the conjecture is true for balanced Boolean functions when $|\rho| \leq \rho_0$. This result was subsequently strengthened by Lei Yu \cite{Yu23}, who confirmed that the conjecture holds for balanced Boolean functions when $|\rho| \leq 0.44$. Li and M\'{e}dard \cite{9272787} investigated the Boolean function that maximizes $\EX{|T_\rho f|^\alpha}$ for a fixed mean, where $\alpha \in [1,2]$. They conjectured that the dictator function is optimal. Subsequently, Barnes and Ozgur in \cite{9174063} showed that the conjecture is related to solving the $\alpha$-Non-Interactive Correlation Distillation (NICD) problem raised by Li and M\'edard in \cite{9272787}.  All of the previous approaches view the problem as a given instance, with parameters given by the channel transition probability and the mean of the Boolean function. In this work, we traverse a path in the space of channels and analyze our objective function along this trajectory. 

It is also known that Conjecture \ref{conj:mi} is implied by the following stronger conjecture, which we call the ``Hellinger Conjecture":
\begin{conj}[\hspace{1sp}\cite{abcjn17}]
    \label{conj:hel}
    Under the same assumptions as in Conjecture \ref{conj:mi}, the following holds:
    \begin{equation}
    \sqrt{1 - \EX{f(\bX)}^2} - \EX{\sqrt{1 - \left((T_\rho f)(\mathbf{Y})\right)^2}} \leq 1 - \sqrt{1 - \rho^2},
    \end{equation}
    where $T_\rho f(\by)=\EX{f(\bX) | \mathbf{Y}=\by}$.
\end{conj}
Recently, \cite{chen2024optimality} found a parametric class of conjectures that interpolates between Conjecture \ref{conj:mi} and Conjecture \ref{conj:hel}. Conjecture \ref{conj:hel} implies a conjecture about the sensitivity of Boolean functions. For a Boolean function \( f \), the sensitivity at a point \( \bx \), represented as \( s_f(\bx) \), is defined as the count of \( \bx \)'s neighbors for which the function produces a value contrary to \( f(\bx) \). Isoperimetric inequalities can help set limits on this sensitivity. For example, the Talagrand isoperimetric inequality demonstrates that for any balanced function, the following holds:
{ \[
\EX{\sqrt{s_f(\bX)}} \geq \frac{1}{\sqrt{2}}.
\]}
Conjecture \ref{conj:hel} implies a strengthened version of this inequality (still unproved) as follows: for any balanced functions, we have
{ \begin{align*}
\EX{\sqrt{s_f(\bX)}} \geq 1.
\end{align*}}
The sensitivity of Boolean functions has received attention in the math community; see \cite{bob97,beltran2023sharp,kahn2019isoperimetric,dir24} for some related results on the sensitivity of Boolean functions. Building on the work in \cite{kahn2019isoperimetric}, it was shown very recently that for all balanced functions \begin{align*}
\EX{s_f^\beta(\bX)} \geq 1,
\end{align*}
for all $\beta \geq 0.50057$.

In this paper, we introduce a differential equation approach to investigate the two conjectures mentioned earlier. We consider a path comprising Binary Symmetric Channels (BSC) with a crossover probability that evolves over time, resulting in a channel output denoted as $\bY_t$. We compute the derivative of our objective function, which represents the mutual information between $f(\bX)$ and the output $\bY_t$ along the path. By establishing bounds on the derivative, we derive new constraints on the endpoints of this path. Observe that this approach is akin to the auxiliary receiver approach in \cite{gon21} but uses a continuum of auxiliary receivers instead of just one auxiliary receiver, i.e., BSC channels whose crossover probability is smaller than that of the given channel. 

Assume that $\bY_t$ and $\bY_{t+\epsilon}$ are two channel outputs at times $t$ and $t+\epsilon$. Instead of single-letterizing the difference between mutual information terms
$I(f(\bX);\bY_t)-I(f(\bX);\bY_{t+\epsilon})$
using the past/future of $\bY_t$ and $\bY_{t+\epsilon}$, we take the derivative of $I(f(\bX);\bY_t)$ with respect to $t$ and then combine the derivative with a natural induction technique on the dimension of the hypercube. This process leads to a functional inequality (on finite dimensions) whose solutions provide requisite lower bounds to the quantity of interest. In particular, backed by numerical simulations, we also conjecture that a particular function satisfies the functional inequality induced from Conjecture \ref{conj:mi}, which, if true, would establish it for balanced functions (and perhaps more).

This paper is organized as follows: Section \ref{sec2} describes the differential equation framework through the example of the most informative Boolean function conjecture. Section \ref{sec3} applies the framework to the Hellinger conjecture. All proofs are given in Section \ref{sec:proofs}.

\textbf{Notation:} We use the following notation in this paper. We use the bold letter $\bX$ to denote the vector of random variables $\bX=(X_1, X_2, \cdots, X_n)$. We use uppercase letter to denote random variables while their values are depicted in lowercase letters. Let $$H_2(x)=x\log_2\frac{1}{x}+(1-x)\log_2\frac{1}{1-x}$$ be the binary entropy function, $$J(x) = \log_2 \frac{1 - x}{x}$$ be the derivative of $H_2(x)$. The inverse function $H_2^{-1}$ will map $[0,1]$ to $[0,\frac12]$. Let $D_2(x\|y)=x\log_2\frac{x}{y}+(1-x)\log_2\frac{1-x}{1-y}$ be the binary $KL$-divergence between the distributions $(x,1-x)$ and $(y,1-y)$.

\section{The general framework of the differential equation approach}
\label{sec2}
Consider the setting in Conjecture \ref{conj:mi}, and let $$p_t = \frac{1 - e^{-2t}}{2}$$
be the crossover probability for some $t \in [0,\infty)$, i.e., $\rho_t=e^{-2t}$. Let 
$$\bY_t=(Y_{t,1},Y_{t,2},\cdots, Y_{t,n})$$ 
be the output of the BSC channel with crossover probability $p_t$. 
Take some arbitrary $P_{F|\bX}$ where $F\in \{-1,1\}$ is a binary random variable such that $P_{F|\bX}(F=-1|\bX=\bx)\in(0,1)$ for all $\bx$. Note that we are excluding the case of $F$ being a function of $\bX$, but the function case can be considered to be a limiting case when the probabilities tend to $0$ and $1$. Let $P_{F,\bX,\bY_t}=P_{F|\bX}P_{\bX,\bY_t}$. Without loss of generality, we can assume that $F\rightarrow \bX\rightarrow \bY_{t_1}\rightarrow \bY_{t_2}$ holds for any $t_1<t_2$. Define $$\gamma(t) = H(F|\bY_t).$$
Note that $\gamma(0)=H(F|\bX)$, $\gamma(\infty)=H(F)$, $I(F;\bX)=\gamma(\infty)-\gamma(0)$ and $I(F;\bY_t)=\gamma(\infty)-\gamma(t)$.
 The definition shows that $\gamma(t)=H(F)-I(F;\bY_t)$ is an increasing function in $t$ because of the data processing inequality. 
 The following lemma (whose proof is given in Section \ref{sec:proofs}) provides the exact derivative.
\begin{lemma}\label{lmm1}
    \begin{align*}
        \frac{d\gamma(t)}{dt} &= \frac{1}{2^{n}}\sum_{\bx\sim \by}(v_{\bx}(t) - v_{\by}(t))(J(v_{\by}(t)) - J(v_{\bx}(t))),
    \end{align*}
    where $v_{\bx}(t) = \emph{Pr}(F = -1|\bY_t = \bx)$ for $\bx\in\mathbb{H}^n$. Here $\bx\sim \by$ stands for the Hamming distance $d_H(\bx,\by) = 1$ and the tuple $(\bx,\by),(\by,\bx)$ are only counted once in the summation.
\end{lemma}
\begin{proof}
We have
    { \begin{align*}
        \gamma(t) &= H(F|\bY_t)= \EX{H(F|\bY_t = \by)}= \EX{H_2(v_\by(t))}= \frac{1}{2^n}\sum_\by H_2(v_\by(t)).
    \end{align*}}
    This implies that
    {\begin{align*}
        \frac{d\gamma(t)}{dt} &= \frac{1}{2^n}\sum_\by J(v_\by(t))\frac{dv_\by(t)}{dt}\\
        &= \frac{1}{2^n}\sum_\by J(v_\by(t)) \lim_{\varepsilon \downarrow 0} \frac{v_{\by}(t+\varepsilon)-v_{\by}(t)}{\varepsilon}\\
        &= \frac{1}{2^n}\sum_\by J(v_\by(t))\lim_{\varepsilon \downarrow 0} \frac{\sum_\bx P(F = -1,\bY_t = \bx|\bY_{t + \varepsilon} = \by) - v_\by(t)}{\varepsilon}\\
        &= \frac{1}{2^n}\sum_\by J(v_\by(t))\lim_{\varepsilon \downarrow 0} \frac{\sum_\bx v_\bx(t)P(\bY_t = \bx|\bY_{t + \varepsilon} = \by) -v_\by(t)}{\varepsilon}\\
        &= \frac{1}{2^n}\sum_\by J(v_\by(t))\lim_{\varepsilon \downarrow 0}\frac{\sum_\bx v_\bx(t)p_\varepsilon^{d_H(\bx,\by)}\left(1-p_\varepsilon\right)^{n - d_H(\bx,\by)} - v_\by(t)}{\varepsilon}\\
        &= \frac{1}{2^n}\sum_\by J(v_\by(t))\lim_{\varepsilon \downarrow 0}\frac{ \sum\limits_{\bx\sim \by}(v_\bx(t)-v_\by(t))p_\varepsilon\left(1-p_\varepsilon\right)^{n - 1}}{\varepsilon} + o(\varepsilon)\\
        &= \frac{1}{2^n}\sum_\by J(v_\by(t))\sum_{\bx\sim \by}(v_\bx(t) - v_\by(t))\\
        &= \frac{1}{2^n}\sum_{\bx\sim \by}(J(v_\by(t)) - J(v_\bx(t)))(v_\bx(t) - v_\by(t)).
    \end{align*}}
\end{proof}
\begin{rem} Note that
    $$(u-w)(J(w)-J(u))=D_2(u\|w)+D_2(w\|u)$$ where $D_2$ is the binary KL divergence.
\end{rem}

\begin{definition}\label{def1}Let $\Psi$ be the class of all non-negative functions $\psi(a,b):(0,1)^2\mapsto \mathbb{R}$ that satisfy the following two conditions:
\begin{itemize}
    \item 
    $\psi(a,b)=0$ when $H_2(a)\leq b$. 
    \item $\psi(1-a,b)=\psi(a,b)$.
    \item Let $P_X$ be an arbitrary distribution on $\mathcal{X}=\{1,2,3,4,5\}$, and $(u_x,w_x)\in(0,1)^2$ for $x\in\mathcal{X}$ be arbitrary. Then, the following 
holds:\begin{align}
        \label{psi:constraint}
        &\frac{1}{2}\mathbb{E}_X[(u_X-w_X)(J(w_X) - J(u_X))] \\
        &\quad\geq \psi\left(\frac{\EX{u_X+w_X}}{2},\frac{\EX{H_2(u_X) + H_2(w_X)}}{2}\right)  - \frac{\psi(\EX{u_X},\EX{H_2(u_X)}) + \psi(\EX{w_X},\EX{H_2(w_X)})}{2}.\nonumber
    \end{align}
\end{itemize}
\end{definition}
The following remark follows from Caratheodery's theorem. See Section \ref{sec:proofs} for details.
\begin{rem} Instead of imposing the third condition in Definition \ref{def1}, i.e. \eqref{psi:constraint}, for $\mathcal{X}$ of cardinality size $5$, we can equivalently require it for  any $P_X$ on a set $\mathcal{X}$ of arbitrary size.\label{rem:card} 
\end{rem}

\begin{thm}\label{thm1}
For every $\psi\in\Psi$ we have  \begin{align}\frac{d\gamma(t)}{dt} \geq \psi(\textnormal{Pr}[F=-1],\gamma(t)).\label{eqnGd}\end{align}
Consequently, if we let
$$g(x)=\int_{\gamma(0)}^{x}\frac{du}{\psi(\textnormal{Pr}[F=-1],u)},\qquad \forall x\geq \gamma(0).$$
we obtain
$$\gamma(t)\geq g^{-1}(t).$$
\end{thm}

\begin{proof}
Take some arbitrary function $\psi(a,b)$. We would like to show that
$$\frac{d\gamma(t)}{dt} \geq \psi(\textnormal{Pr}(F = -1),\gamma(t)).$$

We prove the statement by induction on $n$. For the base case of $n=1$, we have
\begin{align*}
\gamma(t)&=\frac12 H(v_1(t))+\frac12 H(v_{-1}(t))\\
   \textnormal{Pr}(F = -1)&=\frac12 v_1(t)+\frac12 v_{-1}(t),\\
   \frac{d\gamma(t)}{dt} &= \frac{1}{2}(J(v_{1}(t)) - J(v_{-1}(t)))(v_{-1}(t) - v_{1}(t)).
\end{align*}
Thus, we need to show that
\begin{align*}
    &\frac{1}{2}(J(v_{1}(t)) - J(v_{-1}(t)))(v_{-1}(t) - v_{1}(t))\geq \psi\left( \frac12 v_1(t)+\frac12 v_{-1}(t), \frac12 H(v_1(t))+\frac12 H(v_{-1}(t))\right).
\end{align*}
By the first property of $\psi$, we have
$$\psi\left(  v_i(t), H(v_i(t))\right)=0, \qquad i\in\{1,-1\}$$
and we can rewrite the above inequality as
\begin{align*}
    &\frac{1}{2}(J(v_{1}(t)) - J(v_{-1}(t)))(v_{-1}(t) - v_{1}(t))\\&\quad\geq\psi\left( \frac12 v_1(t)+\frac12 v_{-1}(t), \frac12 H(v_1(t))+\frac12 H(v_{-1}(t))\right)-\frac12\psi\left(  v_1(t), H(v_1(t))\right)-\frac12\psi\left(  v_{-1}(t), H(v_{-1}(t))\right).
\end{align*}
The induction basis is established.

Next, assume that the desired inequality holds for $n-1$. We show it for $n$. 
Define $\Tilde{\bY}_t = (Y_{t,1},Y_{t,2},\cdots, Y_{t,n-1})$ to be the subsequence of the first $n-1$ random variables in $\bY_t$. 
Let 
\begin{align*}
    v^+_{\Tilde{\by}}(t) &= v_{(\tilde{\by},1)}(t)=\textnormal{Pr}(F= -1|\Tilde{\bY}_t=\Tilde{\by},Y_{t,n}=1)\\
    v^-_{\Tilde{\by}}(t) &= v_{(\tilde{\by},-1)}(t)=\textnormal{Pr}(F= -1|\Tilde{\bY}_t=\Tilde{\by},Y_{t,n}=-1)
\end{align*}
Let $F_+, F_-$ be two binary random variables, jointly distributed with  with $\Tilde{\bY}_t$ according to conditional laws
$$\textnormal{Pr}(F_+= -1|\Tilde{\bY}_t=\Tilde{\by})=v^+_{\Tilde{\by}}(t)$$
and
$$\textnormal{Pr}(F_-= -1|\Tilde{\bY}_t=\Tilde{\by})=v^-_{\Tilde{\by}}(t)$$
respectively. Letting $\Tilde{\bx},\Tilde{\by}\in\mathbb{H}^{n-1}$ be the first $n-1$ digits of $\bx,\by$ on the $(n-1)$-dimensional subcube, we can write
    { \begin{align}
    \nonumber
        \frac{d\gamma(t)}{dt} &= \frac{1}{2^n}\sum_{\bx\sim \by}(J(v_\by(t)) - J(v_\bx(t)))(v_\bx(t) - v_\by(t))\\\nonumber
        &=  \frac{1}{2^n}\sum_{\Tilde{\bx}}(J(v^+_{\Tilde{\bx}}(t)) - J(v^-_{\Tilde{\bx}}(t)))(v^-_{\Tilde{\bx}}(t) - v^+_{\Tilde{\bx}}(t))
         \\\nonumber
        &\quad+  \frac{1}{2^n}\sum_{\Tilde{\bx}\sim\Tilde{\by}}(J(v^-_{\Tilde{\by}}(t)) - J(v^-_{\Tilde{\bx}}(t)))(v^-_{\Tilde{\bx}}(t) - v^-_{\Tilde{\by}}(t))+\frac{1}{2^n}\sum_{\Tilde{\bx}\sim\Tilde{\by}}(J(v^+_{\Tilde{\by}}(t)) - J(v^+_{\Tilde{\bx}}(t)))(v^+_{\Tilde{\bx}}(t) - v^+_{\Tilde{\by}}(t))\\\label{eqnLS}
        &\geq \frac{1}{2^n}\sum_{\Tilde{\bx}}(J(v^+_{\Tilde{\bx}}(t)) - J(v^-_{\Tilde{\bx}}(t)))(v^-_{\Tilde{\bx}}(t) - v^+_{\Tilde{\bx}}(t)) + \frac{1}{2}\psi\left(\EX{v^+_{\Tilde{\bY}}},\EX{H_2\left(v^+_{\Tilde{\bY}}\right)}\right) + \frac{1}{2}\psi\left(\EX{v^-_{\Tilde{\bY}}},\EX{H_2\left(v^-_{\Tilde{\bY}}\right)}\right)
    \end{align}}
    where the last step follows from the induction hypothesis. Next, from the last property of $\psi$ for the choice of $X=\tilde{\bY}$,
$u_X=v^+_{\tilde{\bY}}$
and
$w_X=v^-_{\tilde{\bY}}$ we obtain
    \begin{align*}
        &\frac{1}{2^n}\sum_{\Tilde{\bx}}(J(v^+_{\Tilde{\bx}}(t)) - J(v^-_{\Tilde{\bx}}(t)))(v^-_{\Tilde{\bx}}(t) - v^+_{\Tilde{\bx}}(t)) \\
        &\quad\geq \psi\left(\frac{\EX{v^+_{\Tilde{\bY}}+v^-_{\Tilde{\bY}}}}{2},\frac{\EX{H_2\left(v^-_{\Tilde{\bY}}\right)+H_2\left(v^+_{\Tilde{\bY}}\right)}}{2}\right) 
         -\frac{1}{2}\psi\left(\EX{v^+_{\Tilde{\bY}}},\EX{H_2\left(v^+_{\Tilde{\bY}}\right)}\right)- \frac{1}{2}\psi\left(\EX{v^-_{\Tilde{\bY}}},\EX{H_2\left(v^-_{\Tilde{\bY}}\right)}\right). 
    \end{align*}
    This equation, along with \eqref{eqnLS} shows that
   \begin{align}
    \nonumber
        \frac{d\gamma(t)}{dt} &= \frac{1}{2^n}\sum_{\bx\sim \by}(J(v_\by(t)) - J(v_\bx(t)))(v_\bx(t) - v_\by(t))
        \\&\geq 
        \psi\left(\frac{\EX{v^+_{\Tilde{\bY}}+v^-_{\Tilde{\bY}}}}{2},\frac{\EX{H_2\left(v^-_{\Tilde{\bY}}\right)+H_2\left(v^+_{\Tilde{\bY}}\right)}}{2}\right)  
        \\&=\psi(\textnormal{Pr}(F = -1),\gamma(t)).
        \end{align}
The proof is complete.
\end{proof}
\begin{lemma}\label{lemma2} $\Psi$ is a non-empty closed convex set, which is also closed under pointwise maximum, i.e., if $\psi_i\in\Psi$ for $i\in\{1,2\}$, then $\psi(a,b)=\max(\psi_1(a,b),\psi_2(a,b))\in\Psi$. Consequently, the class $\Psi$ has maximal element $\psi^*(a,b)$ that pointwise dominates all the other members of $\Psi$. 
\end{lemma}
\begin{proof}
    The proof of this Lemma can be found in Section \ref{sec:proofs}.
\end{proof}

\subsection{A conjecture}
Given a function $\psi$, observe that verifying whether $\psi$ belongs to $\Psi$ requires verifying an inequality with 14 free variables. While this is a finite-dimensional optimization problem, the space of variables is large. However, we can also think of $\psi$ as follows: given a four tuple $(m_u,m_w,e_u,e_w)$ satisfying $H_2(m_u)\geq e_u$ and $H_2(m_w)\geq e_w$, let 
$$\zeta(m_u,m_w,e_u,e_w) := \inf_{(U,W)\in \mathcal{S}}\frac 12 \EX{(U-W)(J(W)-J(U))},$$ where the set $\mathcal{S}$ is the set of all pairs of random variables $(U,W)\in(0,1)^2$ such that 
\[\EX{U}=m_u, \EX{W}=m_w, \EX{H_2(U)}=e_u, \EX{H_2(W)}=e_w.\]
It follows from the definition of $\zeta$ that it is a jointly convex function on four variables, and equation \eqref{psi:constraint} states a lower bound on $\zeta$ in terms of $\psi$:
\begin{align*}&\zeta(m_u,m_w,e_u,e_w)\geq \psi\left(\frac{m_u+m_w}{2},\frac{e_u+e_w}{2}\right)-  \frac 12 \psi(m_u,e_u) -  \frac 12 \psi(m_w,e_w).
\end{align*}
In particular, when $m_u=m$, $m_w=1-m$ and $e_u=e_w=e$, we obtain
\begin{align}
\zeta(m,1-m,e,e)\geq \psi\left(\frac{1}{2},e\right)- \psi(m,e).\label{eqnN}
\end{align}
\begin{thm}\label{thm24}
We have
\begin{align*}
    \zeta(1-m, m, e, e)&=\phi\left(\frac12,e\right)-\phi(m,e)
    \end{align*}
    where $\phi(m,e)$ is defined as follows:
    let
$$\phi(x,y) = \begin{cases}
    \eta(y) - \frac{y}{r}\eta(r)&H_2(x)>y\\
    0& H_2(x)\leq y
\end{cases}$$
where $$\eta(x) = (1-2H_2^{-1}(x))\cdot J(H_2^{-1}(x)),\qquad\forall x\in (0,1]$$
and 
$r\in(0,1]$ is defined as follows: if $x=1/2$, we set $r=1$; else if $x\neq \frac12$, $r\in(0,1)$ is the unique solution of the following equation:
\begin{align}
    \frac{r}{1-2H_2^{-1}(r)}=\frac{y}{|1-2x|}.\label{defreq}
\end{align}

\end{thm}

\begin{proof}
    The proof of this Theorem can be found in Section \ref{sec:proofs}.
\end{proof}

 We make the following conjecture:
\begin{conj}
\label{conj:gue}
    The function $\phi$ belongs to the class $\Psi$. 
\end{conj}

\begin{rem}
Even though we do not have a formal proof of the above conjecture, we can reduce it to inequalities involving at most four variables that can be  verified numerically with confidence. The evidence for the correctness of the above conjecture is detailed in Appendix \ref{appevidence}.
\end{rem}


Next, we show that the conjecture, if true, implies the most informative Boolean function conjecture in the balanced case: 
\begin{lemma}
\label{lem:conjimpbal}
If Conjecture \ref{conj:gue} holds, then Conjecture \ref{conj:mi} holds whenever $F$ is balanced, i.e. $\emph{Pr}(F=-1)=\emph{Pr}(F=1)=\frac12$.
\end{lemma}
\begin{proof}
 First, observe that $\phi(x,y)=\phi(1-x,y)$ and
$$\phi\left(\frac12,y\right)=\eta(y)$$
Then,
{\begin{align}
g(x)&=\int_{\gamma(0)}^{x}\frac{du}{\phi(\frac12,u)}\nonumber
\\&=\int_{\gamma(0)}^{x}\frac{du}{(1-2H_2^{-1}(u))\cdot J(H_2^{-1}(u))}\nonumber
\\&=\int_{H_2^{-1}(\gamma(0))}^{H_2^{-1}(x)}\frac{dt}{1-2t}\label{eqnf}
\\&=\frac12\log\frac{1-2H_2^{-1}(\gamma(0))}{1-2H_2^{-1}(x)}\nonumber
\end{align}}
where in \eqref{eqnf}, we apply the change to the variables 
$u=H_2(t)$. Thus, we get
$$g(\gamma(t))\geq t$$
or
\begin{align}
    e^{-2t}\left[1-2H_2^{-1}(\gamma(0))\right]\geq 1-2H_2^{-1}(\gamma(t)).
\end{align}
or
\begin{align}
    \gamma(t)\geq H_2\left[\frac{1}{2}\left\{1-e^{-2t}\left[1-2H_2^{-1}(\gamma(0))\right]\right\}\right].
\end{align}
Since $\gamma(t)=1-I(F;\bY_t)$ and $p_t = \frac{1 - e^{-2t}}{2}$, we can rewrite the above as 
    \begin{align*}
       1 - I(F;\bY_t)\geq H_2(p_t \ast H_2^{-1}(1 - I(F;\bX)))
   \end{align*}
where $a\ast b = a\cdot(1-b) + (1-a)\cdot b$ is the binary convolution. Given a balanced Boolean function $B$, define $P(F_\epsilon=-1|\mathbf{X}=\mathbf{x}) = 1-\epsilon$, if $B(\mathbf{x})=-1$ and $P(F_\epsilon=1|\mathbf{X}=\mathbf{x}) = 1-\epsilon$, if $B(\mathbf{x})=1$. As $\epsilon \to 0$, we have $I(F_\epsilon;\bX)\to 1$, and $I(F_\epsilon;\bY_t) \to I(B(\bX);Y_t)$, establishing the desired inequality.
\end{proof}
\begin{rem}
    It is possible that a similar statement can be made about unbalanced functions as well. However, in this case, the integral of the reciprocal of $\phi$ does not seem to have a closed form. 
\end{rem}

\begin{rem}
    The following natural guess for $\psi(x,y)$ does not belong to $\Psi$:
$$\psi(x,y)=\eta(1-H_2(x)+y).$$
It has the following counterexample:
let $X\in\{1,2,3\}$ be a ternary random variable with the probability distribution $P_X = (0.1,0.45,0.45)$. Let $(u_1,u_2,u_3) = (0.99,0.9999,0.0001), (w_1, w_2, w_3) = (0.01,0.9999,0.0001)$. One can verify that the above example does not satisfy the inequality \eqref{psi:constraint}.
\end{rem}

\section{The Hellinger Conjecture}\label{sec3}
The Hellinger Conjecture states that
\begin{equation}
    \sqrt{1 - \EX{f(\bX)}^2} - \EX{\sqrt{1 - \left((T_\rho f)(\mathbf{Y})\right)^2}} \leq 1 - \sqrt{1 - \rho^2},
    \end{equation}
    where $T_\rho f(\by)=\EX{f(\bX) | \mathbf{Y}=\by}$. Again, take some arbitrary $P_{F|\bX}$ where $F\in \{-1,1\}$ is a binary random variable such that $P_{F|\bX}(F=-1|\bX=\bx)\in(0,1)$ for all $\bx$. 
Let $\rho_t=e^{-2t}$ and $v_{\bx}(t) = \emph{Pr}(F = -1|\bY_t = \bx)$. Let
$$d_{\bx}(t)\triangleq \EX{F| \mathbf{Y}_t=\bx}=1-2v_{\bx}(t)$$
Define
\begin{align*}
r(t)&=\mathbb{E}_{\bX}{\sqrt{1 - \left(\EX{F| \mathbf{Y}_t=\bX}\right)^2}}=\EX{\sqrt{1-d_{\bm{X}}(t)^2}} =\frac{1}{2^n}\sum_{\bm{x}}\sqrt{1-d_{\bm{x}}(t)^2}.
\end{align*}
Observe that $$r(\infty)=\sqrt{1 - \EX{F}^2}.$$
The function $r(t)$ is increasing and direct calculation shows that its derivative can be calculated as follows:
\begin{lemma}
    We have
    $$r'(t)=\frac{1}{2^{n}} \sum_{(\bm{x},\bm{y}): \bm{x} \sim \bm{y}} (d_{\bm{x}} - d_{\bm{y}})\left(\frac{d_{\bm{x}}}{\sqrt{1-d_{\bm{x}}^2}} - \frac{d_{\bm{y}}}{\sqrt{1-d_{\bm{y}}^2}} \right) .$$
where $\bx\sim \by$ stands for the Hamming distance $d_H(\bx,\by) = 1$ and the tuple $(\bx,\by),(\by,\bx)$ are only counted once in the summation.
\end{lemma}
\begin{proof}
    \begin{align*}
        \frac{d\gamma(t)}{dt} &= \frac{1}{2^n}\sum_\by \frac{-d_\by(t)}{\sqrt{1 - d_\by(t)^2}} \frac{d d_\by(t)}{dv_\by(t)}\frac{dv_\by(t)}{dt}\\
        &= \frac{1}{2^n}\sum_\by \frac{2d_\by(t)}{\sqrt{1 - d_\by(t)^2}}\sum_{\bx\sim\by}(v_\bx(t) - v_\by(t))\\
        &= \frac{1}{2^n}\sum_\by \frac{d_\by(t)}{\sqrt{1 - d_\by(t)^2}}\sum_{\bx \sim \by}(d_\by(t) - d_\bx(t))\\
        &= \frac{1}{2^{n}} \sum_{(\bm{x},\bm{y}): \bm{x} \sim \bm{y}} (d_{\bm{x}} - d_{\bm{y}})\left(\frac{d_{\bm{x}}}{\sqrt{1-d_{\bm{x}}^2}} - \frac{d_{\bm{y}}}{\sqrt{1-d_{\bm{y}}^2}} \right),
    \end{align*}
    where the second step comes from the calculation in Lemma \ref{lmm1}, and the tuple $(\bx,\by),(\by,\bx)$ are only counted once in the summation.
\end{proof}
Let us define the corresponding set for the Hellinger conjecture:
\begin{definition}Let $\Psi_H$ be the class of all non-negative functions $\psi(a,b):(-1,1)\times [0,1)\mapsto \mathbb{R}$ that satisfy the following two conditions:
\begin{itemize}
    \item 
    $\psi(a,b)=0$ when $\sqrt{1-a^2}\leq b$.
    \item $\psi(a,b)=\psi(-a,b)$.
    \item Let $P_X$ be an arbitrary distribution on $\mathcal{X}=\{1,2,3,4,5\}$, and $(u_x,w_x)\in(-1,1)^2$ for $x\in\mathcal{X}$ be arbitrary. Then, the following 
holds:{\begin{align}
\label{psi:constraint2}
        &\frac{1}{2}\mathbb{E}_X\left[(u_X-w_X)\left(\frac{u_X}{\sqrt{1-u_X^2}} - \frac{w_X}{\sqrt{1-w_X^2}}\right)\right] \\&\quad\geq \psi\left(\frac{\EX{u_X+w_X}}{2},\frac{\EX{\sqrt{1-u_X^2} + \sqrt{1-w_X^2}}}{2}\right) - \frac{\psi\left(\EX{u_X},\EX{\sqrt{1-u_X^2}}\right) + \psi\left(\EX{w_X},\EX{\sqrt{1-w_X^2}}\right)}{2}.
    \end{align}}
\end{itemize}
  
\end{definition}

\begin{rem}
\label{rem:hel1}
    The following observations from the previous section carry over almost verbatim:
        $\Psi_H$ is a non-empty closed convex set, which is also closed under pointwise maximum, i.e., if $\psi_i\in\Psi_H$ for $i\in\{1,2\}$, then $\psi(a,b)=\max(\psi_1(a,b),\psi_2(a,b))\in\Psi_H$. Consequently, the class $\Psi_H$ has maximal element $\psi_H^*(a,b)$ that pointwise dominates all the other members of $\Psi_H$. 
\end{rem}

The following lemma follows:

\begin{thm}\label{thm2}
For every $\psi\in\Psi_H$ we have  $$\frac{dr(t)}{dt} \geq \psi\left(\mathbb{E}[F],r(t)\right).$$
Consequently, if we let
$$g(x)=\int_{r(0)}^{x}\frac{du}{\psi(\mathbb{E}[F],u)},\qquad \forall x\geq r(0),$$
we obtain
$$r(t)\geq g^{-1}(t).$$
\end{thm}

\begin{proof}
    The proof of this theorem mimics that of \Cref{thm1} and is omitted.
\end{proof}

\begin{lemma}
\label{lem:twopoint}
The following hold:
\begin{enumerate}[leftmargin=*]
    \item For $(u,w) \in (-1,1)$,
    {\begin{align*}
        &\frac{1}{2}  (u - w)\left(\frac{u}{\sqrt{1-u^2}} - \frac{w}{\sqrt{1-w^2}} \right)= \left(\left(\frac{u-w}{2}\right)^2 + \left(\frac{\sqrt{1-u^2}-\sqrt{1-w^2}}{2}\right)^2\right) \left(\frac{1}{\sqrt{1-u^2}} + \frac{1}{\sqrt{1-w^2}} \right).
    \end{align*}}
    \item For $(u_x,w_x)\in(-1,1)^2$
    { \begin{align*}
        &\frac{1}{2}\mathbb{E}_X\left[(u_X-w_X)\left(\frac{u_X}{\sqrt{1-u_X^2}} - \frac{w_X}{\sqrt{1-w_X^2}}\right)\right]\\&\quad\geq \left(\left(\frac{\EX{u_X-w_X}}{2}\right)^2 + \left(\frac{\EX{\sqrt{1-u_X^2}-\sqrt{1-w_X^2}}}{2}\right)^2\right)\left(\frac{1}{\EX{\sqrt{1-u_X^2}}} + \frac{1}{\EX{\sqrt{1-w_X^2}}} \right).
    \end{align*}}
\end{enumerate}
    
\end{lemma}

\begin{proof}
    The proof of this theorem can be found in Section \ref{sec:proofs}.
\end{proof}

\begin{definition}Let $\hat{\Psi}_H$ be the class of all non-negative functions $\psi(a,b):(-1,1)\times [0,1)\mapsto \mathbb{R}$ that satisfy the following two conditions:
\begin{itemize}
    \item 
    $\psi(a,b)=0$ when $\sqrt{1-a^2}\leq b$. 
    \item $\psi(a,b)=\psi(-a,b)$
    \item For $(a_1,a_2) \in (-1,1)^2$ and $(b_1,b_2)\in [0,1)^2$, we have
    { \begin{align*}
       & \left(\left(\frac{a_1 - a_2}{2}\right)^2 + \left(\frac{b_1-b_2}{2}\right)^2\right)\left(\frac{1}{b_1} + \frac{1}{b_2} \right)  \geq \psi\left( \frac{a_1+a_2}{2},\frac{b_1+b_2}{2}\right) -\frac 12 \psi(a_1,b_1) - \frac 12 \psi(a_2,b_2).
    \end{align*}}
\end{itemize}
\end{definition}

\begin{rem}
    \label{rem:hel2}
    The following points are worth noting:
    \begin{enumerate}
        \item From \Cref{lem:twopoint}, it is immediate that $\hat{\Psi}_H \subseteq \Psi_H$.
        \item We do not have a conjecture for the maximal element of $\Psi_H$, but we have verified that the natural choice
$$\psi(a,b)=\frac{2(1-a^2-b^2)}{b}$$
has counterexamples. Let $X\in\{1,2,3\}$ be a ternary random variable and consider the following setting of parameters:
\begin{align*}P_X &= (0.9118,0.0760,0.0122),\\
(u_1,u_2,u_3) &= (0.9996,0.7316,0.2996),\\
(w_1, w_2, w_3) &= (0.9992,0.1416,0.5866).
\end{align*} 
One can show that this choice violates the constraint in \eqref{psi:constraint2}.
    \end{enumerate}
\end{rem}

\section{Proofs}
\label{sec:proofs}

\subsection{Proof of Remark \ref{rem:card}}
This remark follows from  
Caratheodery's theorem because, given $P_X$ one can find a distribution $Q_X$ with support of size at most
five such that
{\begin{align*}
    \mathbb{E}_{P_X}[(u_X-w_X)(J(w_X) - J(u_X))]&\geq \mathbb{E}_{Q_X}[(u_X-w_X)(J(w_X) - J(u_X))],\\
    \mathbb{E}_{P_X}(w_X)&=\mathbb{E}_{Q_X}(w_X),\\
    \mathbb{E}_{P_X}(u_X)&=\mathbb{E}_{Q_X}(u_X),\\
    \mathbb{E}_{P_X}(H_2(w_X))&=\mathbb{E}_{Q_X}(H_2(w_X)),\\
    \mathbb{E}_{P_X}(H_2(u_X))&=\mathbb{E}_{Q_X}(H_2(u_X)).
\end{align*}}

\subsection{Proof of Lemma \ref{lemma2}}
Clearly $\psi(a,b)=0$ for all $a,b$ belongs to $\Psi$. Therefore, $\Psi$ is non-empty. If $\psi_1,\psi_2 \in \Psi$, it is immediate at $\alpha \psi_1 + (1-\alpha)\psi_2 \in \Psi$ for $\alpha \in [0,1]$, and hence $\Psi$ is convex.

 Let $\psi_n \in \Psi,n\in N$ be a sequence of functions that converge (pointwise) to $\psi_\infty$ on $(0,1)^2$. We have the following pointwise inequalities
    {\begin{align*}
        &\frac{1}{2}\EX{(u_X-w_X)(J(w_X) - J(u_X))}\\
        &\quad\geq \psi_n\left(\frac{\EX{u_X+w_X}}{2},\frac{\EX{H_2(u_X) + H_2(w_X)}}{2}\right) - \frac{\psi_n(\EX{u_X},\EX{H_2(u_X)}) + \psi_n(\EX{w_X},\EX{H_2(w_X)})}{2},
    \end{align*}}
    by taking limits on both side with respect to $n$, we have that
    {\begin{align*}
        &\frac{1}{2}\EX{(u_X-w_X)(J(w_X) - J(u_X))}\\
        &\quad\geq \psi_\infty\left(\frac{\EX{u_X+w_X}}{2},\frac{\EX{H_2(u_X) + H_2(w_X)}}{2}\right) - \frac{\psi_\infty(\EX{u_X},\EX{H_2(u_X)}) + \psi_\infty(\EX{w_X},\EX{H_2(w_X)})}{2},
    \end{align*}}
    Therefore, $\Psi$ is closed.

    On the other hand, $\Psi$ is closed under pointwise maximum. Let $\psi_1,\psi_2\in \Psi$, and set
$\psi(a,b)=\max(\psi_1(a,b),\psi_2(a,b))$. We have that for any $i\in\{1,2\}$:
    {\begin{align*}
        &\frac{1}{2}\EX{(u_X-w_X)(J(w_X) - J(u_X))} \\
        &\quad\geq \psi_i\left(\frac{\EX{u_X+w_X}}{2},\frac{\EX{H_2(u_X) + H_2(w_X)}}{2}\right) - \frac{\psi_i(\EX{u_X},\EX{H_2(u_X)}) + \psi_i(\EX{w_X},\EX{H_2(w_X)})}{2}\\
        &\quad\geq \psi_i\left(\frac{\EX{u_X+w_X}}{2},\frac{\EX{H_2(u_X) + H_2(w_X)}}{2}\right) - \frac{\psi (\EX{u_X},\EX{H_2(u_X)}) + \psi (\EX{w_X},\EX{H_2(w_X)})}{2}.
    \end{align*}}
    By taking the maximum of the right-hand side over $i\in\{1,2\}$, we get that $\psi\in\Psi$. 

\subsection{Proof of Theorem \ref{thm24}}
\subsubsection{Step 1: Existence of minimizer in the expanded space}

We need the following lemma:
\begin{lemma}\label{lemmaMinm}
    The infimum
{\begin{align*}
    \zeta(m_u,m_w,e_u,e_w) = \inf_{p(x),\{u_x,w_x\}}&\frac{1}{2}\EX{(u_X-w_X)(J(w_X) - J(u_X))}
\end{align*}}
subject to $(u_x,w_x)\in(0,1)^2$
\begin{align*}
\EX{u_X} &= m_u,\\
    \EX{w_X} &= m_w,\\
    \EX{H_2(u_X)} &= e_u,\\
    \EX{H_2(w_X)} &= e_w,
\end{align*}
can be written as a minimum 
if we expand the domain of $(u_x,w_x)$ to 
$(u_x, w_x)\in(0,1)^2\cup \{(0,0),(1,1)\}$ and define $(J(w) - J(u))(u - w)=0$ when $(u,w)\in\{(0,0),(1,1)\}$. Moreover, any minimizer must satisfy the following property: if $p(x_1),p(x_2)>0$ then one cannot simultaneously have $u_{x_1}>u_{x_2}$ and $w_{x_1}<w_{x_2}$.
\end{lemma}

\begin{proof}

Fix some alphabet $\mathcal{X}$ and consider  a sequence 
$\{(p_i(x), u_{ix}, w_{ix})\}$ for $i=1,2,\cdots$ converging to the infimum defining $\zeta(m_u,m_w,e_u,e_w)$. Since $\{(p_i(x), u_{ix}, w_{ix})\}\in[0,1]^{3|\mathcal{X}|}$ lies in a compact source, without loss of generality, we can assume that 
$\{(p_i(x), u_{ix}, w_{ix})\}$ converges to some sequence $\{(p(x), u_x, w_x)\}\in[0,1]^{3|\mathcal{X}|}$, where $p(x)$ is a probability distribution on $\mathcal{X}$. Let $\mathcal{S}=\{x\in\mathcal{X}: p(x)>0\}$ be the support set of $X$.

Observe that the limiting $u_x$ and $w_x$ might be equal to $0$ or $1$. Note that $(u-w)(J(w) - J(u))$ converges to infinity if one of $u$ and $v$ converges to $0$ and the other one converges to $1$. Therefore, the only possible case is that the limiting $(u_x,w_x)$ might be equal to $(0,0)$ or $(1,1)$ for some $x$. Note that 
$(J(w) - J(u))(u - w)\geq 0$ for every $u,w\in(0,1)$. Consequently, if $(u(x),w(x))\in\{(0,0),(1,1)\}$ for some $x\in\mathcal{S}$, changing the value of $w_{ix}$ to be equal to $u_{ix}$ for such $x$ would decrease $(J(w_{xi}) - J(u_{xi}))(u_{xi} - w_{xi})$ to $0$, and cannot affect the limit of the other terms. Therefore, if $(u(x),w(x))\in\{(0,0),(1,1)\}$ for some $x\in\mathcal{S}$, without loss of generality we can assume that $u_{ix}=w_{ix}$ for every $i$. Let 
\begin{align*}
    \mathcal{S}_0&=\{x\in\mathcal{S}: u_x=w_x=0\},\\
     \mathcal{S}_1&=\{x\in\mathcal{S}: u_x=w_x=1\},\\
      \mathcal{S}_2&=\{x\in\mathcal{S}: u_x,w_x\in(0,1)\}.
\end{align*}
In this case, the infimum of the expression will be equal to
\begin{align*}
    &\frac{1}{2}\sum_{x\in\mathcal{S}_2} p(x)(J(w_x) - J(u_x))(u_x - w_x).
\end{align*}
This shows that if we extend the domain of $(u,w)$ from $(0,1)^2$ to $(0,1)^2\cup \{(0,0),(1,1)\}$ and define
$(J(w) - J(u))(u - w)= 0$ when $(u,w)\in\{(0,0),(1,1)\}$, the infimum of the expression will also be a minimum and will be attained at some $p(x), u_x, w_x$ where $(u_x, w_x)\in(0,1)^2\cup \{(0,0),(1,1)\}$ for every $x$.

Next, let us consider a minimizer $p(x),u_x,w_x$ for $x\in\mathcal{X}$. Assume that $u_{x_1}>u_{x_2}$ and $w_{x_1}<w_{x_2}$. Then, using the fact that $J(x)$ is a decreasing function, we have
{$$(J(w_{x_1})- J(w_{x_2}))(u_{x_1}-u_{x_2})+ (J(u_{x_1})- J(u_{x_2}))(w_{x_1} -w_{x_2} )>0.$$}

Take $x_1$ and $x_2$ where $p(x_1)\geq p(x_2)>0$. Take a symbol $x_3\notin \mathcal{X}$ and expand the alphabet of $X$ as  $\mathcal{X}'=\mathcal{X}\cup \{x_3\}$. Let $p(X'=x)=p(X=x)$ if $x\notin \{x_1,x_3\}$, $p(X'=x_1)=p(x_1)-p(x_2)$ and $p(X'=x_3)=p(x_2)$. Let $(\hat u_{x'},\hat w_{x'})=(u_{x'},w_{x'})$ if $x'\notin \{x_2,x_3\}$, and 
$(\hat u_{x_2},\hat w_{x_2})=(u_{x_2},w_{x_1})$
and
$(\hat u_{x_3},\hat w_{x_3})=(u_{x_1},w_{x_2})$. Observe that 
{\begin{align*}
    &(u_{x_1}-w_{x_1})(J(w_{x_1})-J(u_{x_1}))+(u_{x_2}-w_{x_2})(J(w_{x_2})-J(u_{x_2}))\\
    &\quad>(u_{x_1}-w_{x_2})(J(w_{x_2})-J(u_{x_1}))+(u_{x_2}-w_{x_1})(J(w_{x_1})-J(u_{x_2}))
\end{align*}}
because the difference between the left-hand side and the right-hand side equals
$$(J(w_{x_1})- J(w_{x_2}))(u_{x_1}-u_{x_2})+ (J(u_{x_1})- J(u_{x_2}))(w_{x_1} -w_{x_2} ).$$
One can then deduce that $X'$ along with $u_{X'}$ and $w_{X'}$ will yield a smaller value of
$\frac{1}{2}\EX{(u_{X'}-w_{X'})(J(w_{X'}) - J(u_{X'}))}$ than $X$, $u_X$ and $w_X$. Moreover, $\mathbb{E}\hat u_{X'}=\mathbb{E}u_{X}$, $\mathbb{E}\hat w_{X'}=\mathbb{E}w_{X}$, $\mathbb{E}H(\hat u_{X'})=\mathbb{E}H(u_{X})$, $\mathbb{E}H(\hat w_{X'})=\mathbb{E}H(w_{X})$. This contradicts the assumption of minimality of $X$, $u_X$ and $w_X$.
\end{proof}

Let us return to the proof of Theorem \ref{thm24}. 
Suppose $p(x), u_x, w_x$ is a minimizer (which exists due to the above lemma when we allow $(0,0)$ and $(1,1)$ in the domain). 

\subsubsection{Step 2: finding a minimizer satisfying $u_x+w_x=1$ for all $x$ where $(u_x,w_x)\in(0,1)^2$ }
We first show that one can find another minimizer satisfying $u_x+w_x=1$ for all $x$ where $(u_x,w_x)\in(0,1)^2$. 
To construct the new minimizer, we first apply a symmetrization argument. Let $\mathcal{X}'=\mathcal{X}\times\{1,2\}$. For any $x'=(x,j)$ where $x\in\mathcal{X}$ and $j\in\{1,2\}$ define
$$p(x')=\frac12 p(x)$$
and define $u_{x'}$ and $v_{x'}$ as follows: 
$$u_{(x,1)}=u_x, \qquad w_{(x,1)}=w_x$$
$$u_{(x,2)}=1-w_x, \qquad w_{(x,2)}=1-u_x$$
One can verify that
\begin{align*}
\EX{u_{X'}} &= m,\\
    \EX{w_{X'}} &= 1-m,\\
    \EX{H_2(u_{X'})} &= e,\\
    \EX{H_2(w_{X'})} &= e.
\end{align*}
Moreover, 
{\begin{align*}
    \EX{(u_X-w_X)(J(w_X) - J(u_X))}=
    \EX{(u_{X'}-w_{X'})(J(w_{X'}) - J(u_{X'}))}
\end{align*}}
Thus, $p(x'), u_{x'}, w_{x'}$ is also a minimizer.

Next, for every $x$ where $(u_x,w_x)\in(0,1)^2$, let $v_x$ be the unique solution of
\begin{align*}
    \frac{H_2(v_x)}{2v_x - 1} = \frac{H_2(u_x) + H_2(w_x)}{2(u_x-w_x)}.
\end{align*}
Here, we set $v_x=1/2$ if $u_x=w_x$. Next, let
$$r_x =\frac{H_2(u_x) + H_2(w_x)}{2H_2(v_x)}.$$
\begin{lemma}\label{lmmnr}
    We have that 
$r_x\in[0,1]$. Moreover, 
\begin{align}
    r_x(1 - 2v_x)J(v_x) \leq \frac{1}{2}(u_x - w_x)(J(w_x) - J(u_x)).\label{eqnRH}
\end{align}
\end{lemma}
The proof of the lemma is given below. We use this lemma as follows: we are choosing the pair
$$u_{(x,1)}=u_x, \qquad w_{(x,1)}=w_x$$
with probability $p(x)/2$ and the pair
$$u_{(x,2)}=1-w_x, \qquad w_{(x,2)}=1-u_x$$
with probability $p(x)/2$. Instead, let us choose the pair
$(0,0)$ with probability $(1-r_x)p(x)/2$, the pair
$(1,1)$ with probability $(1-r_x)p(x)/2$ and finally the pair $(v_x,1-v_x)$ with probability $r_xp(x)/2$. We apply this replacement for every $x$ where $(u_x,w_x)\in(0,1)^2$. One can verify that this transformation preserves $\mathbb{E}[u_{X'}]$, $\mathbb{E}[w_{X'}]$ as well as $\mathbb{E}[H_2(u_{X'})]$, $\mathbb{E}[H_2(w_{X'})]$. Moreover, 
\eqref{eqnRH} implies that the value of the objective function after the transformation is less than or equal to the original value. Observe that after the transformation, we are using either the pairs $(0,0)$, $(1,1)$ or $(v_x,1-v_x)$. This establishes the statement desired in Step 2. 

It remains to prove Proof of Lemma \ref{lmmnr}. 

\begin{proof}[Proof of Lemma \ref{lmmnr}]

 Assume that $u_x\geq w_x$ (the proof for the other case is similar). In this case, $v_x\geq 1/2$. We have
{\begin{align*}
    \frac{H_2(u_x) + H_2(w_x)}{2}&=\frac{H_2(1-u_x) + H_2(w_x)}{2}\leq H_2\left(\frac{1-u_x+w_x}{2}\right).
\end{align*}}
Thus,
\begin{align*}
H_2^{-1}\left(\frac{H_2(u_x) + H_2(w_x)}{2}\right)&\leq \frac{1-u_x+w_x}{2}
\end{align*}
and we obtain
\begin{align*}
    0\leq u_x-w_x&\leq 1-2H_2^{-1}\left(\frac{H_2(u_x) + H_2(w_x)}{2}\right)
\end{align*}
and we can write
{\begin{align*}
    &\frac{H_2(1-v_x)}{1-2(1-v_x)} = \frac{H_2(u_x) + H_2(w_x)}{2(u_x-w_x)}\geq \frac{\frac{H_2(u_x) + H_2(w_x)}{2}}{1-2H_2^{-1}(\frac{H_2(u_x) + H_2(w_x)}{2})}=\frac{H_2(H_2^{-1}(\frac{H_2(u_x) + H_2(w_x)}{2}))}{1-2H_2^{-1}(\frac{H_2(u_x) + H_2(w_x)}{2})}.
\end{align*}}
Since $H_2(x)/(1-2x)$ is increasing in $x\in[0,0.5]$, we get 
$$1-v_x\geq H_2^{-1}(\frac{H_2(u_x) + H_2(w_x)}{2})$$
Thus,
$$H_2(v_x)\geq \frac{H_2(u_x) + H_2(w_x)}{2}$$
This completes the proof for $r_x\in[0,1]$.

We now prove that
\begin{align}
    r_x(1 - 2v_x)J(v_x) \leq \frac{1}{2}(u_x - w_x)(J(w_x) - J(u_x)).\label{ineqalitytoshow}
\end{align}
From the definition of $r_x$ and $v_x$ we have
$$r_x(2v_x-1)=(u_x-w_x).$$
We show the inequality when $u_x\geq w_x$ (which implies $v_x\geq 1/2$). The proof for the other case is similar. Observe that \eqref{ineqalitytoshow} is equivalent to
\begin{align*}
    -J(v_x) \leq \frac{J(w_x) - J(u_x)}{2}.
\end{align*}
Let
\begin{align}
    z &= \frac{\sqrt{w_x(1 - u_x)}}{\sqrt{w_x(1 - u_x)} + \sqrt{u_x(1 - w_x)}}\in\left(0,0.5\right].\label{defZ}
\end{align}
Then, one can verify that 
\begin{align}
    J(z) = \frac{J(w_x) - J(u_x)}{2}\geq 0.\label{defZ2}
\end{align}
We wish to show that $-J(v_x)=J(1-v_x)\leq J(z)$. Since $J$ is a decreasing function, this is equivalent to $1-v_x\geq z$. Since $H_2(x)/(1-2x)$ is increasing in $x\in[0,0.5]$, this inequality is equivalent to
\begin{align*}
    \frac{H_2(z)}{1-2z}\leq \frac{H_2(1-v_x)}{1-2(1-v_x)} = \frac{H_2(u_x) + H_2(w_x)}{2(u_x-w_x)}.
\end{align*}
To sum this up, for $z$ defined by \eqref{defZ}, we need to show the above inequality. Equivalently,
{\begin{align*}
    &H_2(u_x) + H_2(w_x)\\
    &\geq 2\left(\sqrt{w_x(1 - u_x)} + \sqrt{u_x(1 - w_x)}\right)^2H_2\left(\frac{\sqrt{w_x(1 - u_x)}}{\sqrt{w_x(1 - u_x)} + \sqrt{u_x(1 - w_x)}}\right)\\
    &= -w_x(1-u_x)\log_2(w_x(1 - u_x)) -u_x(1-w_x)\log_2(u_x(1 - w_x)) \\
    &\quad- 2\sqrt{w_x(1-u_x)u_x(1-w_x)}\log_2\sqrt{w_x(1-u_x)u_x(1-w_x)} \\
    &\quad+\left(\sqrt{u_x(1-w_x)} + \sqrt{w_x(1-u_x)}\right)^2\log_2\left(\sqrt{u_x(1-w_x)} + \sqrt{w_x(1-u_x)}\right)^2.
\end{align*}}
By expanding the left-hand side, the inequality is equivalent to proving
{\begin{align}
    &2\sqrt{w_x(1-u_x)u_x(1-w_x)}\log_2\sqrt{w_x(1-u_x)u_x(1-w_x)}\nonumber \\
    &\quad\geq w_xu_x\log_2(w_xu_x) + (1-w_x)(1 - u_x)\log_2((1 - w_x)(1-u_x))\nonumber\\
    &\qquad+ \left(\sqrt{u_x(1-w_x)} + \sqrt{w_x(1-u_x)}\right)^2\log_2\left(\sqrt{u_x(1-w_x)} + \sqrt{w_x(1-u_x)}\right)^2.
\label{eqnDf}
\end{align}}
Let 
$$a=\frac{\frac{u_x}{1-u_x}+\frac{w_x}{1-w_x}}{2},$$ and $$b=\sqrt{\frac{u_x}{1-u_x}\frac{w_x}{1-w_x}}.$$
Observe that we have $a \geq b\geq 0$. If we divide both sides of \eqref{eqnDf} by $(1-u_x)(1-w_x)$ and apply a change of variable, the inequality will be equivalent to 
\begin{align*}
   & (1+2a+b^2) \log_2 (1+2a+b^2)  - 2(a+b) \log_2(2a+2b)\geq 2b^2 \log_2b - 2b \log_2(b).
\end{align*}
Taking the derivative in $a$, we obtain
$2\log(1+2a+b^2) - 2\log(2a+2b) \geq 0. $ Therefore, it suffices to prove the above inequality for $a=b$. In other words, we need to show that
{\begin{align*}
   &  2b \log(b) + (1+2b+b^2) \log (1+2b+b^2) - 2b^2 \log b - 4b \log4b \geq 0.
\end{align*}}
Equivalently, we need to show that
{\begin{align*}
   &  f(b):=(1+b)^2 \log (1+b) - b^2 \log b -  b \log (b) - 2b \log4 \geq 0.
\end{align*}}
Observe that
$b^2 f\left(\frac{1}{b}\right) = f(b)$. Therefore, it suffices to consider $b \in (0,1]$.
\end{proof}

Note that
\begin{align*}
    f'(b) &= 2(1+b)\log(1+b) - 2b\log b - \log b - 2 \log 4.
\end{align*}
One can verify that there is exactly one solution for $f''(b)=0$ in $(0,1)$.
Further, one can observe that $f(b)$ initially increases and is concave, attains a maximum, then decreases, and then turns convex in the interval $(0,1)$. Note that $f'(1)=0$ and hence decreasing as $b \uparrow 1$. As $f(0)=0$ and $f(1)=0$, and $f'(0)=\infty$, if $f(b)$ had a local minimum with a negative value, it must have had two local maximums in $(0,1)$. This contradicts that there is exactly one solution for $f''(b)=0$ in $(0,1)$. Therefore, $f(b) \geq 0$ for $b \in (0,1)$, and equality only holds when $b \in \{0,1\}.$

\subsubsection{Step 3: Simplifying the minimizer}
Assume that we have a minimizer satisfying $u_x+w_x=1$ for all $x$ where $(u_x,w_x)\in(0,1)^2$. Take $x_1$ and $x_2$ such that $p(x_1), p(x_2)>0$ and  $u_{x_1}+w_{x_1}=1$ and $u_{x_2}+w_{x_2}=1$. We claim that $u_{x_1}=u_{x_2}$ and $w_{x_1}=w_{x_2}$. Assume that $u_{x_1}> u_{x_2}$. Then, by Lemma \ref{lemmaMinm}, we get $w_{x_1}\geq w_{x_2}$. This would imply $1-u_{x_1}\geq 1-u_{x_2}$ or $u_{x_1}\leq u_{x_2}$ which is a contradiction. Thus, $u_{x_1}=u_{x_2}$ and $w_{x_1}=w_{x_2}$. This shows that we may put weights on at most a single term $(v,1-v)$. Thus, we can take the pairs $(0,0)$ with some probability $p_0$, $(1,1)$ with some probability $p_1$ and $(v,1-v)$ for some $v$ with some probability $p_2$. Since $\mathbb{E}[u_X]=m$ and $\mathbb{E}[w_X]=1-m$, we get
\begin{align}
    vp_2+p_1&=m\\
(1-v)p_2+p_1&=1-m
\end{align}
Summing up, we get, $p_2+2p_1=1$ which implies $p_0=p_1=(1-p_2)/2$ and we get
\begin{align}
    vp_2+(1-p_2)/2&=m.\end{align}
Using 
$$e=p_2H_2(v)$$
the objective function reduces to the expression given in the statement of the theorem.

\subsection{Proof of \Cref{lem:twopoint}}
We wish to show that
{\begin{align*}
        &\frac{1}{2}  (u - w)\left(\frac{u}{\sqrt{1-u^2}} - \frac{w}{\sqrt{1-w^2}} \right)= \left(\left(\frac{u-w}{2}\right)^2 + \left(\frac{\sqrt{1-u^2}-\sqrt{1-w^2}}{2}\right)^2\right)\left(\frac{1}{\sqrt{1-u^2}} + \frac{1}{\sqrt{1-w^2}} \right).
    \end{align*}}
    Let $\sqrt{1-u^2}=x$ and $\sqrt{1-w^2}=y$. Then, the desired equality can be written as, requiring to show that,
    \begin{align*}
        \frac{1}{2}  (u - w)\frac{(uy-wx)}{xy} = \left(\left(\frac{u-w}{2}\right)^2 + \left(\frac{x-y}{2}\right)^2\right)\frac{(x+y)}{xy}.
    \end{align*}
    Using $u^2+x^2=1$ and $w^2+y^2=1$, the desired equality can be written as, requiring to show that,
    \begin{align*}
          (u - w)(uy-wx) = (1-uw-xy)(x+y).
    \end{align*}
    which is immediate by expansion and using $u^2+x^2=1$ and $w^2+y^2=1$.

    Cauchy-Schwarz implies that if $Y> 0$, then $\EX{\frac{X^2}{Y}} \geq \frac{\EX{X}^2}{\EX{Y}}. $ Now, the second part is immediate, by applying the earlier part and expanding the expression, naturally, into four terms. 

\section{Discussion and Conclusion}
The most-informative Boolean function conjecture and Talagrand's isoperimetric inequality are well-studied questions in information theory and combinatorics, respectively. The latter is a limiting case of the Hellinger conjecture. Despite a naturally degraded structure with respect to the noise operator, previous studies have not explicitly used this observation. In this work, we study both conjectures via the lens of the degraded channels and work on the derivative of the quantity of interest. This approach yields a finite-dimensional functional inequality whose solutions provide a dimension-independent lower bound to the original problem.

\section*{Acknowledgements}
The authors would like to thank Lau Chin Wa for discussions at the preliminary stages of the work, as well as Venkat Anantharam with whom they had shared some initial forms of the conjecture.

\thispagestyle{empty}
\bibliographystyle{IEEEtran}
\bibliography{mybiblio}

\thispagestyle{empty}

\appendices
\section{Evidence for Conjecture \ref{conj:gue}}
\label{appevidence}

Define the odd function $L: [0,\frac 12) \to \mathbb{R}_+$ according to $$L(u)=\frac{2H_2(u)}{1-2u},\qquad\qquad u \neq \frac 12$$ and $L(u)=\infty$ when $u=\frac 12$. Then, $r$ in \eqref{defreq} can be expressed as $r=H_2\left(L^{-1}\left(\frac{2y}{|1-2x|}\right)\right)$. Moreover, for every $x,y$ satisfying $H_2((1-x)/2)\geq y$, we have
$$xJ\left(L^{-1}\left(\frac{y}{x}\right)\right)=\phi(\frac12,\frac{y}{2})-\phi(\frac{1-x}{2},\frac{y}{2})=\zeta\left(\frac{1-x}{2},\frac{1+x}{2},\frac{y}{2},\frac{y}{2}\right).$$
Next, define
$$\kappa(u,w):=(u-w)\frac{J(w) - J(u)}{2}-|u-w|J\left(L^{-1}\left(\frac{H_2(u) + H_2(w)}{|u - w|}\right)\right),$$
for $(u,w)\in(0,1)^2$, $u\neq w$. When $u=w\in[0,1]$, we set $\kappa(u,w)=0$. 

Since $$\zeta(u,w,H_2(u),H_2(w))=\zeta(1-w,1-u,H_2(w),H_2(u))=(u-w)\frac{J(w) - J(u)}{2}$$
 by the joint convexity of $\zeta$, we obtain
\begin{align}
\kappa(u,w)&=\frac12\zeta(u,w,H_2(u),H_2(w))+\frac12\zeta(1-w,1-u,H_2(w),H_2(u))\nonumber\\&\qquad-\zeta\left(\frac{1+u-w}{2},\frac{1+w-u}{2},\frac{H_2(u)+H_2(w)}{2},\frac{H_2(u)+H_2(w)}{2}\right)\nonumber
\\&\geq 0. \label{eq:kapnonneg}
\end{align}

We propose the following inequalities (backed by numerical evidence).
\begin{conj}\label{conj2} The following holds:
    \begin{itemize}
        \item for $u,w\in(0,0.5]$ we have $\kappa(u,w)\leq \kappa(1-u,w)$,
        \item $\kappa(H_2^{-1}(u),H_2^{-1}(w))$ is jointly convex in $(u,w)$.
    \end{itemize}
\end{conj}

\begin{rem}
The following are worth noting:
\begin{itemize}
  \item By plotting the two-dimensional surface of \( \kappa(u,w) - \kappa(1-u,w) \) for \( u, w \in (0, 0.5) \), one observes that the the difference becomes zero along the line $u=w$ and $u=\frac 12$, and seems to be strictly positive elsewhere.
  \item By plotting the surface of $\frac{\partial^2}{\partial u^2} \kappa(H_2^{-1}(u),H_2^{-1}(w))$, one observes that it is strictly non-negative.
  \item By plotting the surface of the determinant of the Hessian on $\kappa(H_2^{-1}(u),H_2^{-1}(w))$, one also observes that it is positive except along the line $u=w$. But along the line $u=w$, we know that $\kappa(H_2^{-1}(u),H_2^{-1}(w))=0$ and achieves its minimum (implying local convexity).
  \end{itemize}
\end{rem}

\begin{thm}\label{thmap}
    If Conjecture \ref{conj2} holds,    then the following lower bound on $\zeta$ holds:
\begin{align*}
\zeta(m_u,m_w,e_u,e_w)\geq&
    \phi\left(\frac{1-|H_2^{-1}(e_u)-H_2^{-1}(e_w)|}{2},\frac{e_u + e_w}{2}\right)+
    (H_2^{-1}(e_u)-H_2^{-1}(e_w))\frac{J(H_2^{-1}(e_w)) - J(H_2^{-1}(e_u))}{2}\\&-\phi\left(\frac{1-|m_u-m_w|}{2},\frac{e_u + e_w}{2}\right).
\end{align*}

\end{thm}
\begin{proof}
    [Proof of Theorem \ref{thmap}]     
    For every $u\in[0,1]$, let $u^*=H_2^{-1}(H_2(u))$, i.e.,
    $$u^*=\begin{cases}
        u&\text{if }u\leq 1/2,\\
        1-u&\text{if }u> 1/2.
    \end{cases}$$
    Then, by the first assumption and the fact that $\kappa(u,w)=\kappa(1-u,1-w)$ for every $u,w$ we obtain that 
$$\kappa(u,w)\geq \kappa(u^*,w^*).$$
Therefore, for every $P_{U,W}$ we can write
\begin{align}
   \mathbb{E}[\kappa(U,W)]&\geq \mathbb{E}[\kappa(U^*,W^*)]\\&=\mathbb{E}[\kappa(H_2^{-1}(H(U^*)),H_2^{-1}(H(W^*)))]\\&\geq \kappa(H_2^{-1}(\mathbb{E}H(U^*)),H_2^{-1}(\mathbb{E}H(W^*))), 
\end{align}
where the last step follows from convexity of $\kappa(H_2^{-1}(u),H_2^{-1}(w))$ (conjectured in \Cref{conj2}). 
From the definition of $\kappa$, the above inequality reduces to
\begin{align*}
    \mathbb{E}\left[(U-W)\frac{J(W) - J(U)}{2}\right]&\geq \mathbb{E}\left[|U-W|J\left(L^{-1}\left(\frac{H_2(U) + H_2(W)}{|U - W|}\right)\right)\right]+\kappa(H_2^{-1}(\mathbb{E}H_2(U^*)),H_2^{-1}(\mathbb{E}H_2(W^*)))
    \\&\geq 
    |\mathbb{E}[U]-\mathbb{E}[W]|J\left(L^{-1}\left(\frac{\mathbb{E}[H_2(U)] + \mathbb{E}[H_2(W)]}{|\mathbb{E}[U] - \mathbb{E}[W]|}\right)\right)+\kappa(H_2^{-1}(\mathbb{E}H_2(U^*)),H_2^{-1}(\mathbb{E}H_2(W^*)))
\end{align*}
where the last step follows from the joint convexity of the function (see Lemma \ref{corap1} in Appendix \ref{appendixB}),
$$(x,y)\mapsto |x|J(L^{-1}\left(\frac{y}{|x|}\right).$$
Overall, this yields the following lower bound on $\zeta$:
$$\zeta(m_u,m_w,e_u,e_w)\geq |m_u-m_w|\cdot J(L^{-1}(\frac{e_u + e_w}{|m_u - m_w|}))+\kappa(H_2^{-1}(e_u),H_2^{-1}(e_w)).$$
One can verify that this lower bound can be also expressed as stated in the statement of the theorem.
\end{proof}

The following conjecture (along with Theorem \ref{thmap})  establishes that the function $\phi$ belongs to the class $\Psi$:
\begin{conj}
    For any $m_u,m_w,e_u,e_w\in[0,1]$ where $H_2(m_u)\geq e_u$ and $H_2(m_w)\geq e_w$ we have \begin{align*}
        &\phi\left(\frac{1-|H_2^{-1}(e_u)-H_2^{-1}(e_w)|}{2},\frac{e_u + e_w}{2}\right)+
    (H_2^{-1}(e_u)-H_2^{-1}(e_w))\frac{J(H_2^{-1}(e_w)) - J(H_2^{-1}(e_u))}{2}\\&-\phi\left(\frac{1-|m_u-m_w|}{2},\frac{e_u + e_w}{2}\right)\geq \phi\left(\frac{m_u+m_w}{2},\frac{e_u+e_w}{2}\right)-  \frac 12 \phi(m_u,e_u) - \frac 12 \phi(m_w,e_w).
    \end{align*}
    \label{newconj}
\end{conj}
Since $\phi(x,y)=\phi(1-x,y)$, the inequality remains invariant under the transformation $m_u\leftrightarrow 1-m_u$ (and similarly under the transformation $m_w\leftrightarrow 1-m_w$). Therefore, without loss of generality we can assume that $m_u,m_w\in[0,0.5]$.

Observe that the inequality in Conjecture \ref{newconj} involves only four free variables and holds by extensive simulations.

{\color{blue}
\section{Updates to the original version}

Extensive numerical simulations on Conjectures 4 and 5 have been performed by several independent groups.
\begin{itemize}
    \item A group of undergraduate students, Aya Jmoud and Houssameddine Abous, studying at EPFL, have verified conjectures 4 and 5 for some regimes of parameters $(m_u, m_w, e_u, e_w)$. Please see their upload on github, \hyperlink{https://github.com/housss77/Boolean-MI-Conjecture-Proofs}{https://github.com/housss77/Boolean-MI-Conjecture-Proofs}.
    \item Conjecture 5 does not hold for all parameters, and hence is false as stated above. In particular, it fails at some points near the boundary.
    \begin{itemize}
        \item Using an internally developed Gemini-based agentic system at Google, Adel Javanmard, Honghao Lin, Vahab Mirrokni, and David Woodruff came up with the following counterexamples (and more similar ones) around April 22nd 2026, to Conjecture \ref{newconj}:
        \begin{align*}
    m_u &= 3.13350269308102\times 10^{-10}\\
    m_w &= 3.00850269308102\times 10^{-10}\\
    e_u &= 1.005 \times 10^{-8}\\
    e_w &= 9.95 \times 10^{-9}.
\end{align*}
For this counterexample, the gap between the left-hand-side and the right-hand-side of Conjecture \ref{newconj} is roughly $-1.655 \times 10^{-14}$.
Another counterexample is of the form
\begin{align*}
    m_u &= 0.999999999675730277814977853268\\
    m_w &= 2.99269693032512\times 10^{-10}\\
    e_u &= 1.01 \times 10^{-8}\\
    e_w &= 9.9 \times 10^{-9}.
\end{align*}
For this counterexample, the gap between the left-hand-side and the right-hand-side of Conjecture 5 is roughly $-6.552\times 10^{-14}$.
\item An independent researcher, Fan Zhou, reached out to the authors on July 27, 2026 with another counterexample to Conjecture \ref{newconj},
\begin{align*}
    m_u &= 9.983459207843347 \times 10^{-8}\\
    m_w &= 2.2335357818271923\times 10^{-7}\\
    e_u &= 2.4657561145651218\times 10^{-6}\\
    e_w &= 3.458243821002324\times 10^{-6}.
\end{align*}
The gap for this counterexample is roughly $-7.155 \times 10^{-10}$.
    \end{itemize}
\end{itemize}

\subsection{Status of this approach}
The counterexamples above do not disprove Conjecture \ref{conj:gue}, the main component of our program. Conjecture \ref{conj:gue} posits that
\[ \zeta(m_u,m_w,e_u,e_w) \geq \phi\left(\frac{m_u+m_w}{2},\frac{e_u + e_w}{2} \right) - \frac{1}{2}\phi(m_u,e_u) - \frac{1}{2}\phi(m_w,e_w). \]
We are currently exploring (in collaboration with Javanmard, Lin, Mirrokni, and Woodruff) promising modifications to complete the proof of Conjecture \ref{conj:gue}.
In fact, we have recently developed some other lower bounds on $\zeta(m_u,m_w,e_u,e_w) $ that show that Conjecture \ref{conj:gue} continues to hold for all the known counterexamples, including those above. 
 
}


\section{Some properties of $L(x)$}\label{appendixB}

\begin{thm}The following hold:
    \label{th:knwnprop}
    \begin{enumerate}
    \item $x\mapsto J(L^{-1}(x))$ is convex over $[0,\infty)$,
    \item 
    $x\mapsto \frac{1}{x}J(L^{-1}(x))$ is  decreasing for $x>0$.
    \end{enumerate}
\end{thm}
\begin{proof}
Define 
$$N(x)\triangleq J'(x)=\frac{1}{x(x-1)\ln(2)}$$
\begin{enumerate}
    \item
    Note that we have
    \begin{align*}
        \frac{dL^{-1}(x)}{dx} &= \frac{1}{L'(L^{-1}(x))},\\
        \frac{d^2L^{-1}(x)}{dx^2} &= -\frac{L''(L^{-1}(x))}{(L'(L^{-1}(x)))^3}.
    \end{align*}
    Therefore,
    \begin{align*}
        \frac{d J(L^{-1}(x))}{dx}&= \frac{N(L^{-1}(x))}{L'(L^{-1}(x))}\\
        \frac{d^2 J(L^{-1}(x))}{dx^2}&= \frac{N'(L^{-1}(x))}{(L'(L^{-1}(x)))^2}-\frac{N(L^{-1}(x))L''(L^{-1}(x))}{(L'(L^{-1}(x)))^3}.
    \end{align*}
    Let $u = L^{-1}(x) \in \left[0,\frac{1}{2}\right]$, this is equivalent for us to show that
    \begin{align*}
        \frac{N'(u)}{L'(u)^2} - \frac{N(u)L''(u)}{L'(u)^3} \geq 0.
    \end{align*}
    Note that since $L(u) = \frac{H_2(u)}{\frac{1}{2} - u}, \forall u \in \left[0,\frac{1}{2}\right]$, we have 
    \begin{align*}
        L'(u) &= \frac{-2\ln(u) - 2\ln(1-u)}{(1 - 2u)^2\ln2}\\
        L''(u) &= \frac{-8u(1-u)(\ln u + \ln(1-u)) - 2(1-2u)^2}{(1 - 2u)^3u(1-u)\ln2}\\
        N(u) &= -\frac{1}{u(1-u)\ln 2}\\
        N'(u) &= \frac{1}{u^2\ln2} - \frac{1}{(1 - u)^2 \ln2}.
    \end{align*}
    Note that $L'(u) \geq 0$, we have that it is equivalent to showing that
    \begin{align*}
        N'(u)L'(u) - N(u)L''(u)\geq 0.
    \end{align*}
    Substituting for $N(u), N'(u)$, this is equivalent to showing
    \begin{align*}
        (1 - 2u)L'(u) + u(1-u)L''(u)\geq 0.
    \end{align*}
    Substituting $L'(u),L''(u)$, this is equivalent to showing
    \begin{align*}
        -\frac{\ln u + \ln(1-u)}{1-2u} - \frac{4u(1-u)(\ln u + \ln(1-u)) + (1-2u)^2}{(1 - 2u)^3} \geq 0.
    \end{align*}
    Since $u \leq \frac{1}{2}$, this is equivalent to
    \begin{align*}
        -\ln u - \ln(1-u) - \frac{4u(1-u)(\ln u + \ln(1-u))}{(1 - 2u)^2} \geq 1,
    \end{align*}
    which can be further written as
    \begin{align*}
        -\ln u - \ln(1-u) \geq 1 - 4u + 4u^2.
    \end{align*}
    Let $t(u) = -\ln u - \ln(1 - u)  -4u^2 + 4u - 1$, we have that
    \begin{align*}
        t'(u) &= -\frac{1}{u} + \frac{1}{1 - u} - 8u + 4\\
        &= \frac{(2u - 1)^3}{u(1-u)},
    \end{align*}
    which is monotone decreasing on $\left[0,\frac{1}{2}\right]$. Note that the inequality is true for $u = \frac{1}{2}$, by monotonicity, we have that $J(L^{-1}(x))$ is convex on $[0,\infty)$.
    \item Note that
    \begin{align*}
        \frac{d \frac{1}{x}J(L^{-1}(x))}{dx}&= \frac{N(L^{-1}(x))}{xL'(L^{-1}(x))} - \frac{J(L^{-1}(x))}{x^2}.
    \end{align*}
    Let $u = L^{-1}(x) \in \left[0,\frac{1}{2}\right]$, to show that the function is monotone decreasing, since $L'(u) \geq 0$, this is equivalent to showing that
    \begin{align*}
        L(u)N(u) - J(u)L'(u) \leq 0.
    \end{align*}
    On the other hand $L(u) \geq 0, L'(u) \geq 0, J(u) \geq 0$, and $N(u) \leq 0$ for $u \in (0,\frac 12)$, implying the above immediately.
    \end{enumerate}
\end{proof}
\begin{lemma} \label{corap1}The function
    \begin{align}
       (x,y)\mapsto |x|J\left(L^{-1}\left(\frac{y}{|x|}\right)\right)\label{lxy} 
    \end{align}
is jointly convex on $\mathbb{R}\times (0,\infty)$.
\end{lemma}
\begin{proof}
    Since $J(L^{-1}(x))$ is convex, the corresponding perspective function
\begin{align}
    xJ\left(L^{-1}\left(\frac{y}{x}\right)\right)
\end{align}
is jointly convex for $x,y>0$. Since $x\mapsto \frac{1}{x}J(L^{-1}(x))$ is  decreasing for $x>0$, for every fixed $y>0$, $x\mapsto xJ(L^{-1}(y/x))$ is  increasing in $x$. Thus, 
\begin{align*}
    \left|\frac{x_1+x_2}{2}\right|J\left(L^{-1}\left(\frac{\frac{y_1+y_2}{2}}{\left|\frac{x_1+x_2}{2}\right|}\right)\right)&\leq \frac{|x_1|+|x_2|}{2}J\left(L^{-1}\left(\frac{\frac{y_1+y_2}{2}}{\frac{|x_1|+|x_2|}{2}}\right)\right)
    \\&\leq \frac12|x_1|J\left(L^{-1}\left(\frac{y_1}{|x_1|}\right)\right)+\frac12|x_2|J\left(L^{-1}\left(\frac{y_2}{|x_2|}\right)\right).
\end{align*}
    This shows that \eqref{lxy} is jointly convex.
\end{proof}

\end{document}